\documentclass[11pt]{amsart}

\usepackage{blindtext}
\usepackage[paperwidth=8in, paperheight=10in]{geometry}
\usepackage{fullpage}
\usepackage{amsmath,amsthm,amssymb,latexsym,color}

\usepackage{amsfonts}

\usepackage{xcolor}

\usepackage{mathrsfs} 

\usepackage{hyperref}

\newtheorem{thm}{Theorem}[section]
\newtheorem{theorem}{Theorem}[section]
\newtheorem{prop}[thm]{Proposition}
\newtheorem{lem}[thm]{Lemma}
\newtheorem{cor}[thm]{Corollary}

\newcommand{\PP}{\mathbb{P}}
\newcommand{\E}{\mathbb{E}}
\newcommand{\R}{\mathbb{R}}

\newcommand{\brk}[1]{\left( #1\right)}
\newcommand{\set}[1]{\left\{#1\right\}}
\newcommand{\abs}[1]{\left| #1 \right|}

\newcommand{\wt}[1]{{\rm wt}\brk{#1}}
\newcommand{\wtB}[1]{{\rm wt}\Big(#1\Big)}
\def\F{{\mathbb F}_q}
\def\pfp{{\rho}}
\def\Event{{\mathcal E}}

\title{Distribution of the minimal distance of random linear codes}
\author{Jing Hao\and Han Huang\and Galyna Livshyts\and Konstantin Tikhomirov
}
\address{School of Mathematics,
Georgia Institute of Technology,
686 Cherry street,
Atlanta, GA 30332.}

\definecolor{bg}{RGB}{38,50,56}
\begin{document}

\maketitle

\begin{abstract} In this paper,
we study the distribution of the minimal distance (in the Hamming metric) of a random linear code of dimension $k$ in
$\mathbb{F}_q^n$. We provide quantitative
estimates showing that the distribution function of the minimal distance is close ({\it{}superpolynomially} in $n$)
to the cumulative distribution function
of the minimum of $(q^k-1)/(q-1)$ independent
binomial random variables with parameters $\frac{1}{q}$ and $n$.
The latter, in turn, converges to a Gumbel distribution at integer points
when $\frac{k}{n}$ converges to a fixed number in $(0,1)$.
Our result confirms in a strong sense that apart from identification of the weights of proportional codewords,
the probabilistic dependencies introduced
by the linear structure of the random code, produce a negligible effect on the minimal code weight.
As a corollary of the main result, we obtain an improvement of the Gilbert--Varshamov bound for $2<q<49$.
\end{abstract}

\section{Introduction}

Let $\mathbb{F}_q$ be a finite field.
A {\it linear code $C$} is a subspace of $\mathbb{F}_q^n$ where $n$ is the length of the code.
The parameter $q$ of the field is referred to as the alphabet size.
The {\it size} of $C$ is the number of elements in $C$.
For a (not necessarily linear) code with size $M$, alphabet size $q$, and length $n$, the {\it
information rate} $R$ is defined to be $\log_q(M)/n$. For a linear code this number is equal to $k/n$, where $k$ is the dimension of the code as a vector space. 

Another fundamental parameter is the relative minimal distance.
Let {\it the Hamming distance} between any two codewords $u=(u_1, \cdots, u_n)$ and $v=(v_1, \cdots, v_n)$ 
in $\mathbb{F}_q^n$ be given by 
\[
  d(u,v) := |\{1 \le i \le n,\: u_i \neq v_i\}|,
\]
and {\it the Hamming weight} of a codeword $u$ be defined as $\wt{u}:= d(u,0)$.
For linear codes, the minimal distance between two distinct codewords in a code
is equal to the minimal weight over all nonzero codewords.
It is well-known that a code with a minimal distance $d$ can correct up to $\frac{d-1}{2}$ errors.
The {\it relative minimal distance} $\delta$ is defined as the ratio $\frac{d}{n}$.

In coding theory, the trade-off between the code rate $R$ and error-correcting ability $\delta$ is a central topic of study. 
Let $q$ be fixed.
For linear codes, 
Manin \cite{Manin1982} has proved that there exists a function $\alpha_q(\cdot)$ with the following property:
for any $\delta_0\in (0,1-1/q)$ and any $R_0\leq \alpha_q(\delta_0)$, there is an infinite sequence of linear codes
with the relative minimal distance converging to $\delta_0$ and the rate converging to $R_0$;
on the other hand, for every $R_0> \alpha_q(\delta_0)$, such a sequence does not exist. 
An explicit description of $\alpha_q(\cdot)$ remains a major open problem (see \cite{Goppa,Jiang2004,Joyner2011},
as well as \cite{Mceliece1977} for an upper bound for 
$\alpha_q$).
Considerable work has been done to obtain explicit constructions for linear codes with
good rate and relative minimal distance (we refer, in particular, to \cite{Tsfasman1982}).

Rather than considering special codes, one may be interested in studying the statistical properties on the space
of all linear codes, using probabilistic methods.
A classical result in this direction is the {\it Gilbert--Varshamov argument}.
Gilbert \cite{Gilbert1952} and \cite{Varshamov1957} independently gave upper bound for the size of a
(not necessarily linear) code given $n$ and $d$.
Let $A_q(n,d)$ be the maximal size of a code of length $n$ over $\mathbb{F}_q$ and with minimum distance $d$.
Then,
\begin{align} \label{eq:GVbound}
    A_q(n,d) \ge \frac{q^{n-1}}{\sum_{j=0}^{d-2} {n \choose j}(q-1)^j},
\end{align}
and, moreover, there are {\it linear} codes that can achieve this bound i.e.\
there exists a linear code over $\mathbb{F}_q$
with dimension at least 
$n - \lfloor{\log_q \sum_{j=0}^{d-2}{ {n \choose j}}(q-1)^j}\rfloor - 1 $.
The proof of the result can be obtained by a union bound argument.

Recall that {\it the $q$-ary entropy function} is defined by
\[
	H_q(x) :=  x \log_q(q-1) - x \log_q(x) - (1-x)\log_q(1-x).
\]
In \cite{Gallager1963}, it was shown that for $q=2$ and given a rate $R_0$ and $\varepsilon > 0$,
the probability that a random linear code of length $n$ and rate $R_0$, uniformly distributed on the set of linear codes of
the given length and rate,
has the minimal distance $ d < n(\delta_0 - \varepsilon)$,
is exponentially small in $n$.
Here $0<\delta_0< \frac{1}{2}$ is the solution of the equation $R_0=1 - H_2(\delta_0)$.  
On the other hand, if we fix any $\delta_0$ satisfying $0 \le \delta_0 < 1- \frac{1}{q}$
and $0 < \varepsilon \le  1 - H_q(\delta_0)$,
then the Gilbert--Varshamov argument implies that there exist infinitely many linear codes
with a rate $R \ge 1 - H_q(\delta_0) - \varepsilon$.
By taking $\varepsilon \rightarrow 0$, one would obtain a lower bound for the function $\alpha_q(\delta)$ mentioned above:
\[
	\alpha_q(\delta) \ge 1- H_q(\delta).
\]
In fact, as was proved in \cite{Barg2002}, the following law of large numbers holds for the minimal distance of a sequence
of random linear codes:
if $n\to\infty$ and the rate $k/n$ converges to a number $R_0\in(0,1)$ then the
relative minimal distance converges (almost surely) to the number $\delta_0$ given by the equation $R_0=1-H_q(\delta_0)$.
Moreover, the probability that a random linear code of length $n$ has the relative minimal distance outside
of the interval $[\delta_0-\varepsilon,\delta_0+\varepsilon]$, is exponentially small in $n$
(we remark here that in the same paper it was shown that the minimal distance of random non-linear codes
is asymptotically worse than in the linear setting).

Our goal in this paper is to obtain a more precise description of the distribution of the minimal distance of
random linear codes. The main statement is the following
\begin{theorem}\label{main}
For any prime power $q$ and any real numbers $R_1<R_2$ in $(0,1)$ there is $c(R_1,R_2,q)>0$ with the following property.
Let positive integers $k,n$ satisfy $R_1\leq k/n\leq R_2$, and let
$\mathcal C$ be the random linear code uniformly distributed on the set of all linear codes in $\F^n$ of dimension $k$.
Denote by $F_{\rm dmin}$ the cumulative distribution function of the minimal distance of $\mathcal C$.
Further, let $w_{\min}$ be the minimal weight of $\frac{q^k-1}{q-1}$ i.i.d.\ uniform random vectors in $\F^n$,
and $F_{\rm wmin}$ be its cumulative distribution function. Then
$$
\sup\limits_{x\in\R}\big|F_{\rm dmin}(x)-F_{\rm wmin}(x)\big|=O\big(\exp(-c(R_1,R_2,q)\,\sqrt{n})\big).
$$
\end{theorem}
A surprising feature of this result is that the distribution of the minimal distance can be approximated by a c.d.f.\
of the minimum of i.i.d.\ binomial variables with precision superpolynomial in $n$. In a sense, this result asserts that dependencies
between codeword weights
introduced by the linear structure of the code, produce a negligible effect on the distribution of the minimal weight.

The proof of the result is based on analysis of moments of certain functionals associated with the code.
We remark that in a recent work by Linial and Mosheiff \cite{Linial2018}, the authors
calculated centered moments for number of codewords of a random linear code with a given weight. The approach used in that paper influenced our work.

As an immediate corollary of our result, we obtain the following statement which gives an $\Theta(n^{1/2})$ improvement over the
classical Gilbert--Varshamov bound:
\begin{cor}\label{cor: mp GV}
For any prime power $q$, any $\alpha\in(0,1)$, any integer $n$, and $d\in[\alpha n,(1-\alpha)(n-n/q)]$
there is a linear code of size at least
$$
c n^{1/2}\,\frac{q^n}{\sum_{j=0}^{d-1} {n \choose j}(q-1)^j},
$$
where $c>0$ may only depend on $\alpha$ and $q$.
\end{cor}
We note that existence of {\it non-linear} codes of size at least
$
cn\,\frac{q^n}{\sum_{j=0}^{d-1} {n \choose j}(q-1)^j}
$
has been previously established in \cite{Jiang2004,Vu2005}.
Linear double-circulant binary codes beating the Gilbert--Varshamov bound were considered in \cite{Gaborit2008}.
To our best knowledge, the above improvement for $2<q<49$ is new.

%

\medskip

Further, we obtain an explicit limit theorem for the distribution of the minimal distance. Due to the discrete nature of our random variable,
the convergence to a Gumbel distribution can only be established on the points along certain arithmetic progressions:
\begin{theorem}[The limit theorem for the minimal distance]
  \label{thm:limiting}
Let $q$ be a prime power, and let $R_1<R_2$ be numbers in $(0,1)$.
Let $(k_n)$ be a sequence of positive integers such that $R_1\leq k_n/n\leq R_2$
for all large $n$. For any $n$ let $d_{\min}(n)$ be the minimal distance of the random linear code uniformly
distributed on the set of linear codes of length $n$ and dimension $k_n$. Further, for any $n$ let $d_0(n)$ be the largest integer
satisfying
$$u(n):=\frac{q^{k_n}-1}{q-1}\sum_{i=0}^{d_0(n)} {n\choose i}\Big(1-\frac{1}{q}\Big)^i\,q^{i-n}\leq 1.$$
Denote by $\xi_n$ the random variable
$$
\xi_n:=\big(d_0(n)-d_{\min}(n)\big)\log \frac{(q-1)(n-d_0(n))}{d_0(n)}-\log u(n).
$$
Then, as $n\to\infty$, we have
$$
\sup\limits\bigg\{|\PP\{\xi_n<t\}-G(t)|:\;t\in \log \frac{(q-1)(n-d_0(n))}{d_0(n)}\, \mathbb{Z}-\log u(n)\bigg\}\longrightarrow 0,
$$
where $G$ is the Gumbel law given by $G(t)=e^{-e^{-t}}$.
\end{theorem}

%


\bigskip

The paper is organized as follows.
In Section \ref{sec: auxiliary results for the binomial distribution}, 
we consider some auxiliary results for the binomial distribution,
including a limiting result for the minimum of i.i.d.\ binomial random variables.  
At the end of the section, we show how the main result of the paper implies Theorem~\ref{thm:limiting}.

In Section \ref{sec: moments comparison}, we consider the set of random vectors $\{Y_a :\; a \in \mathbb{F}_q^k \backslash \{0\}\}$ uniformly distribtued on $\mathbb{F}_q^n$ that are mutually independent up to the constraint that $Y_a = Y_b$ whenever $a$ and $b$ are proportional.  We study moments of the random variable that counts number of codewords with weights less than or equal to $d$ in this configuration as well as that of random linear code ensemble and give a quantitative comparison between them.

Finally, in Section \ref{sec: analysis of the distribution of the minimal distance} we give the comparison of the c.d.f.\ of minimum distance between these two ensembles.  Due to the discrete nature of this problem, either c.d.f.\
can be obtained by solving a set of linear equations involving quantities we computed in previous sections.
Then we give a quantitative comparison by estimating the truncation errors and moment differences.


\bigskip

{\bf Acknowledgment.} The authors would like to thank Prof.\ Alexander Barg for valuable suggestions.
G.L.\ is supported by NSF grant CAREER DMS-1753260.
K.T.\ is supported by the Sloan Research Fellowship.

\section{Auxiliary results for the binomial distribution}
\label{sec: auxiliary results for the binomial distribution}

Our goal in this section is to obtain quantitative estimates for the 
distribution of the minimum of i.i.d.\ binomial random variables (with specially chosen parameters).
Although the material of this section is rather standard, we prefer to include it in the exposition for the reader's convenience.

Let $1\leq m\leq (q-1)^n$ and let $X_1,\dots,X_{m}$ be i.i.d.\ vectors uniformly distributed in 
$\F^n$.
Here, we are interested in estimates of the quantities
$$
\PP\big\{\min\limits_{i\leq m}\wt{X_i}\leq d\big\},\quad d\geq 0,
$$
where $\wt{X_i}$ is the number of non-zero components of $X_i$.
Denote
$$
\pfp_{d}:=\PP\big\{\wt{X_1}\leq d\big\}=\sum_{i=0}^d {n\choose i}\Big(1-\frac{1}{q}\Big)^i\,q^{i-n}.
$$

\medskip

We start by recording the following approximations to $\pfp_{d}$:

\begin{prop}\label{prop: 09860986}
    For any $\alpha\in (0,1)$ there is $C_\alpha>0$ with the following property.
    Assume that $n\geq 1$ and $C_\alpha {\log(n) } \le d\le (1-\alpha)(1-1/q)n$. 
    Then we have 
\begin{align}\label{equation1}
\frac{\pfp_{d}}{{n \choose d}q^{-n}(q-1)^d}= 
  \big(1+O_\alpha(\log n/n)\big)\frac{n-d+1}{n-\brk{\frac{q}{q-1}}d+1}.
\end{align}
Furthermore, for any positive integer $t\leq\sqrt{d}$, we have  
\begin{equation}\label{eq: pdnpdtn}
   \frac{\pfp_{d+t}}{\pfp_{d}} 
    =\left(1+O_{\alpha}\left(\frac{\log n}{n}
    +\frac{t^2}{d}\right)\right)
    \left(\frac{(q-1)(n-d)}{d}\right)^{t}.
       \end{equation}
\end{prop}
\begin{proof}
We have
\begin{align*}
 \frac{\pfp_{d}}{{n \choose d}q^{-n}(q-1)^d}  &= 
 1 + \frac{1}{q-1}\frac{d}{n-d+1}+\left(\frac{1}{q-1}\right)^2\frac{d}{n-d+1}\frac{d-1}{n-d+2} + \dots\\
 &\hspace{1cm}+\left(\frac{1}{q-1}\right)^d\frac{d(d-1)\cdots 1}{(n-d+1)(n-d+2)\cdots n}\\
 &\le \frac{1}{1-\frac{d}{(n-d+1)(q-1)}}\\
&= {\frac{n-d+1}{n-\brk{\frac{q}{q-1}}d+1}}.
\end{align*}
On the other hand, for any positive integer $t\leq d$ we have 
\begin{align*}
  \frac{\pfp_{d}}{{n \choose d}q^{-n}(q-1)^d} &\ge 
  1 +  \frac{1}{q-1}\frac{d}{n-d+1}+ \left(\frac{1}{q-1}\right)^2\frac{d}{n-d+1}\frac{d-1}{n-d+2} 
  + \dots\\
&\hspace{1cm}+ \left(\frac{1}{q-1}\right)^t\frac{d(d-1)\cdots (d-t+1)}{(n-d+1)(n-d+2)\cdots (n-d+t)}\\
&  \ge \frac{1-\big(\frac{d-t+1}{(n-d+t)(q-1)}\big)^t}{1-\frac{d-t+1}{(n-d+t)(q-1)}}\\
& = \brk{1-\Big(\frac{d-t+1}{(n-d+t)(q-1)}\Big)^t}
\frac{n-d+1+(t-1)}{n-\frac{q}{q-1}(d-t+1)+1}.
\end{align*}

Observe that 
\begin{align*}
  \frac{n-d+1+(t-1)}{n-\frac{q}{q-1}(d-t+1)+1}
  &= \frac{n-d+1+(t-1)}{n-d+1}
  \frac{n-(\frac{q}{q-1})d+1}{n-\frac{q}{q-1}d+1+\frac{q}{q-1}(t-1)}\,\,
  \frac{n-d+1}{n-\brk{\frac{q}{q-1}}d+1}\\
  &= 
    \brk{1+\frac{t-1}{n-d+1}} 
    \brk{1+\frac{\frac{q}{q-1}(t-1)}{n-(\frac{q}{q-1})d+1}}^{-1}
  \frac{n-d+1}{n-\brk{\frac{q}{q-1}}d+1}.
\end{align*}
With $d\le (1-\alpha)(1-\frac{1}{q})n$, we have
$$
  n-d+1 > n-\Big(\frac{q}{q-1}\Big)d+1 \ge \alpha n.
$$
Thus, if $t\leq \alpha n$, we obtain 

{

\begin{align}\label{eq:ratio1}
  \frac{n-d+1+(t-1)}{n-\frac{q}{q-1}(d-t+1)+1}= 
  \brk{
    1+ O\Big(\frac{t}{\alpha n}\Big)
  }\frac{n-d+1}{n-\brk{\frac{q}{q-1}}d+1}.
\end{align}
}
Taking $t=-\frac{\log n}{\log \brk{1-\frac{1}{2}\alpha}}$, and using the assumption on $d$, we get for all large enough $n$
{
\begin{align*}
  \brk{\frac{d-t+1}{(n-d+t)(q-1)}}^t &\le 
  \brk{
    \frac{(1-\alpha)(1-\frac{1}{q})n +1}
    {(n-(1-\frac{1}{q})n)(q-1)}
  }^t
  \le \brk{1-\frac{1}{2}\alpha}^t
  = \frac{1}{n}.
\end{align*}
}
Combining the above, we get \eqref{equation1}. 

\medskip

Next, observe that for $t\le n-d$ we have 
\begin{align}\label{binom-assympto}
  \left(\frac{n-d-t+1}{d+t}\right)^t  {n \choose d} 
  \le {n \choose d+t} 
  \le \left(\frac{n-d}{d+1}\right)^t  {n \choose d}
\end{align}
and 

{ 
\begin{align*}
  \frac{n-d-t+1}{d+t} &= \frac{n-d-t+1}{n-d}\,\frac{d+1}{d+t}\,\frac{n-d}{d+1}\\
  &= \Big(1-\frac{t-1}{n-d}\Big)\Big(1-\frac{t-1}{d+t}\Big)\frac{n-d}{d+1}\\
  &= \Big(1+O_\alpha\Big(\frac{t}{d+t}\Big)\Big) \frac{n-d}{d+1}.
\end{align*}
Hence, when $t\leq\sqrt{d}$, we have 
\begin{align*}
  \frac{{n \choose d+t} }{{n \choose d} } 
  =  \Big(1+O_\alpha\Big(\frac{t^2}{d}\Big)\Big)\brk{\frac{n-d}{d+1}}^t.
\end{align*}
Combining this with \eqref{equation1} and \eqref{eq:ratio1}, we get 
$$
 \frac{\pfp_{d+t}}{\pfp_{d}}=\brk{\frac{(q-1)(n-d)}{d}}^{t}\Big(1+O_\alpha\Big(\frac{\log n}{n}+\frac{t^2}{d}\Big)\Big).
$$
}

\end{proof}

\bigskip



The next proposition provides an approximation of the minimum of independent binomial variables in terms of the Gumbel distribution.
Although the computations seem to be rather standard, we prefer to include them for reader's convenience.

\begin{prop}\label{gumbleprop} Fix $q\ge 2$ and $\alpha\in(0,1)$. Let $q^{\alpha n}\leq m\leq q^{(1-\alpha)n}$ and
  let $d_0$ be the largest integer such that $\rho_{d_0}m\leq 1$. 
    Let $X_1,\dots,X_m$ be i.i.d.\ binomial random variables with 
    parameters $n$ and $\frac{1}{q}$, i.e.\
  $$\PP\{X_j=a\}={n \choose a}\Big(1-\frac{1}{q}\Big)^a q^{a-n},\quad a=0,1,\dots,n,$$
  and set $Y:=\min_{j=1,\dots,m}X_j.$ 
  Then
  $$
  \PP\left\{Y-d_0> \frac{-t}{{\log \frac{(q-1)(n-d_0)}{d_0}}}
  -\frac{\log (\rho_{d_0} m)}{\log \frac{(q-1)(n-d_0)}{d_0}}\right\}
  =o_{\alpha,q}(1)+\exp\big(-e^{-t}\big),
  $$
  for all $t\in \log \frac{(q-1)(n-d_0)}{d_0} \mathbb{Z}-\log(\rho_{d_0}m)$.
  \end{prop}
  
  \begin{proof}
  Without loss of generality, we can assume that $n\geq C_{\alpha,q}$ for a large constant $C_{\alpha,q}$
  depending on $\alpha,q$.
  Observe that there is $c\in(0,1)$ depending only on $\alpha$ and $q$ such that the condition on $m$ implies $c n\leq d_0\leq (1-c)(n-n/q)$. 
  For any integer $d$ we have
  $$
  \PP\left\{Y\leq d\right\}=1-\left(1-\PP\{X_1\leq d\}\right)^m,
  $$
  where in our notation,
  $$\PP\{X_1\leq d\}=\pfp_d.$$
  Thus,
  $$
  1-\PP\left\{Y\leq d\right\}=\exp\big(-(1+o_{\alpha,q}(1))\rho_d m\big)
  =\exp\Big(-(1+o_{\alpha,q}(1))\rho_{d_0} m\;\frac{\rho_d}{\rho_{d_0}}\Big),
  $$
  whenever { $d$ is an integer less equal to $(1-c/2)(n-n/q)$}, so that $\rho_d=o_{\alpha,q}(1)$.
  Further, applying the second assertion of 
  Proposition \ref{prop: 09860986}, we get that for 
  any integer $d$ with $|d-d_0|= o(\sqrt{d_0})$,
  $$
  \frac{\rho_d}{\rho_{d_0}}
      =\left(1+o_{\alpha,q}(1)\right)
      \left(\frac{(q-1)(n-d_0)}{d_0}\right)^{d-d_0}.
  $$
  Hence,
  $$
  1-\PP\left\{Y\leq d\right\}=
  \exp\bigg(-(1+o_{\alpha,q}(1))\rho_{d_0} m\;\left(\frac{(q-1)(n-d_0)}{d_0}\right)^{d-d_0}\bigg),\quad |d-d_0|= o(\sqrt{d_0}).
  $$
  Further, note that by our conditions on $m$ we have $\rho_{d_0} m\geq \tilde c$ for some $\tilde c>0$ depending only on $\alpha$ and $q$.
  Moreover, $\frac{(q-1)(n-d_0)}{d_0}>\frac{1}{1-c}$, by the above obvervation for $d_0$.
  Thus, we can write
  $$
  1-\PP\left\{Y\leq d\right\}=o_{\alpha,q}(1)+\exp\bigg(-\rho_{d_0} m\;\left(\frac{(q-1)(n-d_0)}{d_0}\right)^{d-d_0}\bigg)
  $$
  for {\it all} integers $d$, or, in other form,
  $$
  1-\PP\left\{Y-d_0\leq \frac{t}{{\log \frac{(q-1)(n-d_0)}{d_0}}}
  -\frac{\log (\rho_{d_0} m)}{\log \frac{(q-1)(n-d_0)}{d_0}}\right\}
  =o_{\alpha,q}(1)+\exp\big(-e^{t}\big),\quad 
  $$
  for all $t\in \log \frac{(q-1)(n-d_0)}{d_0} \mathbb{Z}+\log(\rho_{d_0}m)$.
  The result follows.
  \end{proof}





It is not difficult to see that the above proposition and the main theorem of the paper imply Theorem \ref{thm:limiting}.

%

\section{Moments comparison for proportional codes}
\label{sec: moments comparison}

Fix $a\in\R^k$ and $d\geq 0$. Given the independent random vectors $X_1,\dots,X_k$ uniform on $\F^n$,
we define
$$
Z_d:=\sum_{a\in \F^k \backslash \set{0}}{W_a}(d),\quad d\geq 0,
$$
where $W_a(d)$ is the indicator of the event 
$$\Big\{\wtB{\sum_{i=1}^k a_i X_i}\le d\Big\}.$$

For any $a,b\in \F^k \backslash \set{0}$,
we say $a$ and $b$ are proportional if there exists $f \in \F\backslash \set{0}$ 
such that $a=f\,b$ (here the multiplication is in the field $\F$).
Notice that if $a$ and $b$ are proportional, then, $\sum_{i=1}^k a_iX_i$ and $\sum_{i=1}^k b_iX_i$ are proportional
as well.
In particular, the supports of the linear combinations are the same, and thus $W_a(d)=W_b(d)$ whenever $a$ and $b$ are proportional.

Let $\set{Y_a}_{a\in \F \backslash \set{0}}$ 
be random vectors uniformly distributed on $\F^n$ and mutually independent
up to the constraint that $Y_a = Y_b$ whenever $a$ and $b$ are proportional. Define
$$
   \widetilde{Z_d} := \sum_{a\in \F^k \backslash \set{0}}
    \widetilde{W_a}(d)
$$
where $\widetilde{W}_a(d)$ is the indicator function 
of the event $\{\wt{Y_a}\le d\}$.

The goal of this section is to compare the moments of $\widetilde{Z}_d$ and $Z_d$ assuming certain constraints on the parameters $n,k$ and $d$.
The main statement of the section is
\begin{prop}\label{prop: main sec 2}
    For any $\lambda_0 \in (0,1)$ there are $c_{\text{\tiny{\ref{prop: main sec 2}}}}(\lambda_0, q)>0$
and $C_{\text{\tiny{\ref{prop: main sec 2}}}}(\lambda_0, q)>0$ with the following property.
    Suppose $d,n\in \mathbb{N}$ satisfy
    $\frac{d}{n} \leq \lambda_0 (1-\frac{1}{q})$, and $d^2/n^{3/2}\geq C_{\text{\tiny{\ref{prop: main sec 2}}}}(\lambda_0, q)$.
    Then for any positive integer $m\le c_{\text{\tiny{\ref{prop: main sec 2}}}}(\lambda_0, q)d^2/n^{3/2}$ such that
    $q^k\rho_d\geq \exp\big(-\frac{c_{\text{\tiny{\ref{prop: main sec 2}}}}(\lambda_0, q)d^4}{n^3 m}\big)$, we have 
    $$
        \E {Z_d}^m = \big(1+O(\exp(-c_{\text{\tiny{\ref{prop: main sec 2}}}}(\lambda_0, q)d^4/n^3))+O(2^{-k/2})\big) \,
        \E\widetilde{Z_d}^m.
    $$
\end{prop}
Before proving the proposition we need to make some preparatory work.

Let $\ell\leq m\le k$ be positive integers. Suppose $I_1,\dots, I_\ell$ is a partitioning of $[m]$ into non-empty set.
Denote by $\Omega(I_1,\dots,I_\ell)$ the collection of all $m$--tuples
$(a^1,\dots, a^m)\in \brk{\F^k\backslash\set{0}}^m$
such that $a^i$ is proportional to $a^j$ if and only if $i,j\in I_t$ 
for some $t\in [\ell]$.
Further, define
\begin{equation}\label{eq: omega def}
\Omega_\ell:=\big\{(v^1,\dots, v^\ell) \in \brk{\F^k\backslash\set{0}}^\ell:\;\mbox{$v^1,\dots, v^\ell$ are pairwise non-collinear}\big\}.
\end{equation}
Note that 
there is a natural $(q-1)^{m-\ell}$--to--one mapping from $\Omega(I_1,\dots,I_\ell)$ onto
$\Omega_\ell$ which assigns $(a_{\min\{j\in I_t\}})_{t=1}^\ell$ to each $(a^1,\dots, a^m)$.

Now, in view of the above remarks,
\begin{align*}
  Z_d^m = \sum_{a^1,\dots, a^m \in \F^{k}\backslash\set{0}} \prod_{i=1}^m W_{a^i}(d)
    =  \sum_{\ell=1}^m\sum_{I_1,\dots,I_\ell} 
    \sum_{\,v^1,\dots, v^\ell \in  \Omega_\ell }
    \brk{q-1}^{m-\ell}\prod_{i=1}^\ell W_{v^i}^{\abs{I_i}}(d),
\end{align*}
where the second summation is taken over all partitions $I_1,\dots,I_\ell$ of $[m]$ into non-empty sets.
Notice that $W_{v^i}^{\abs{I_i}}(d)=W_{v^i}(d)$, so we can simplify the above representation to
\begin{align}\label{eq: -98529587-593}
  Z_d^m=\sum_{\ell=1}^m S(m,\ell)\brk{q-1}^{m-\ell}
  \brk{\sum_{\,v^1,\dots, v^\ell \in  \Omega_\ell }
  \prod_{i=1}^\ell W_{v^i}(d)},
\end{align}
where $S(m,\ell)$ is the number of ways to partition $[m]$ into $\ell$ non-empty sets
(a {\it{}Stirling number of the second kind}).
The above formula works for $\widetilde{Z}_d^m$ as well, up to replacing $W_{v^i}(d)$
with $\widetilde W_{v^i}(d)$.

\medskip

The central technical statement of the section is the following
\begin{prop} \label{prop: expectationFixedRank}
    For any $\lambda_0 \in (0,1)$ there are $c_{\text{\tiny{\ref{prop: expectationFixedRank}}}}(\lambda_0, q)>0$
and $C_{\text{\tiny{\ref{prop: expectationFixedRank}}}}(\lambda_0, q)>0$ with the following property.
  Suppose $d,n\in \mathbb{N}$ satisfy
  $\frac{d}{n} \leq \lambda_0 (1-\frac{1}{q})$ and $d\geq C_{\text{\tiny{\ref{prop: expectationFixedRank}}}}(\lambda_0, q)$. 
  Suppose further that $s\leq k$, and $\brk{v^1,\,v^2,\,\dots,\,v^{s}}$ are linearly independent vectors in $\F^k$, and that 
  $ v^{s+1} = \sum_{i=1}^s c_iv^i$ 
  for some $c_i \in \F \backslash \set{0}$. Then
  $$
    \E \prod_{i=1}^{s+1} W_{v^i}(d) \le C\rho_{d}^s \,\exp(-c_{\text{\tiny{\ref{prop: expectationFixedRank}}}}(\lambda_0, q)d^4/n^3)
  $$
  where $C>0$ is a universal constant.
\end{prop}
Let us postpone the proof, and proceed with the argument.
As a corollary of the above statement, we have 

\begin{cor}\label{cor: rank cor}
  Suppose $d,n\in \mathbb{N}$ are as in Proposition~\ref{prop: expectationFixedRank}.
  Suppose further that $\ell\leq k$, and $v^1,\,v^2,\,\dots,\,v^{\ell}$ are non-zero vectors in $\F^k$ such that
  the rank of $\brk{v^1,\,v^2,\,\dots,\,v^{\ell}}$ is $r<\ell$. Then
  $$
  \E \prod_{i=1}^{\ell} W_{v^i}(d) \le C\rho_{d}^r \,\exp(-c_{\text{\tiny{\ref{prop: expectationFixedRank}}}}(\lambda_0, q)d^4/n^3).
$$
\end{cor}
\begin{proof}
    By rearranging the indices, we may assume there exists $s$ 
    such that the vectors $v^1,\dots, v^s,$ $v^{s+2},\dots, v^{r+1}$
    form a linearly independent set, and $v^{s+1}=\sum_{i=1}^s c_iv^i$
    where $c_i\in \F\setminus\{0\}$.
    Next,
    $$
        \E \prod_{i=1}^{\ell} W_{v^i}(d) \le \E \prod_{i=1}^{r+1} W_{v^i}(d)
        = \E \prod_{i=1}^{s+1}W_{v^i}(d) \prod_{j=s+2}^{r+1}W_{v^j}(d)
        = \E \prod_{i=1}^{s+1}W_{v^i}(d) \rho_{d}^{r-s},
    $$
    and the result follows by applying Proposition~\ref{prop: expectationFixedRank}.
\end{proof}

Next, we need to estimate cardinality of the set of $\ell$--tuples of vectors
$\brk{v^1,\,v^2,\,\dots,\,v^{\ell}}\in \Omega_\ell$ with a given rank $r$. 
\begin{lem}\label{lem: mtuplesrank}  For $r \le \ell \le k$, denote
$$
\Omega_{r,\ell}:=\set{\brk{v^1,\,v^2,\,\dots,\,v^\ell}\in \Omega_\ell \,:\, 
       {\rm dim}\brk{ 
        {\rm span}\brk{v^1,\,\dots,\,v^\ell} 
        }= r
      }.
$$
Then
    $$
      |\Omega_{r,\ell}|\le { \ell \choose r} q^{r(\ell-r)}\prod_{i=0}^{r-1}\brk{q^k-q^i}.
    $$
    When $r=\ell$, equality holds, implying
    $$
      \frac{|\Omega_{r,\ell}|}{|\Omega_{\ell,\ell}|} \le { \ell \choose \ell- r} 
      \frac{\brk{q^r}^{\ell-r}}{\prod_{i=r}^{\ell-1}\brk{q^k-q^i}}. 
    $$
  \end{lem}
  \begin{proof}
    Let $I\subset [\ell]$ with $\abs{I}=r$. 
    We will consider vectors $(v^1,\dots, v^\ell)$ such that  
    $\set{v^i}_{i\in I}$ forms a linearly independent set,
    and $v^j$ lies in the span of $\set{v^i}_{i\in I}$ 
    for each $j\notin I$. 
    Without loss of generality, we will assume $I=[r]$.
    The cardinality of the set of all $r$--tuples $(v^1,\dots, v^r)$ of linearly independent vectors is 
    $$
      \brk{q^k - 1}\brk{q^k-q}\brk{q^k-q^2}\cdots 
      \brk{q^k - q^{r-1}}
    $$
    where $q^k -q^{i-1}$ represents the number of choices
    of $v^i$ which is not a linear combination of $v^1,\dots, v^{i-1}$
    with the assumption that $v^1,\dots, v^{i-1}$ are linearly 
    independent. Having chosen $v^1,\dots, v^r$, the vectors
    $v^{r+1},\dots v^\ell$ are vectors in the span of $v^1, \dots, v^r$, which has cardinality $q^{r}$. 
    Therefore, there are at most
    $
       \brk{q^r}^{\ell-r}\prod_{i=0}^{r-1}\brk{q^k-q^i}
    $
    choices for $\brk{v^1,\dots, v^\ell}$ satisfying 
    the above conditions. We obtain the desired bound 
    because there are $\ell \choose r$ subsets of $[\ell]$ with
    cardinality $r$.
  \end{proof}

\medskip

Now we are ready to prove Proposition~\ref{prop: main sec 2}.
\begin{proof}[Proof of Proposition~\ref{prop: main sec 2}]
We set
$$c_{\text{\tiny{\ref{prop: main sec 2}}}}(\lambda_0, q):=
\frac{c_{\text{\tiny{\ref{prop: expectationFixedRank}}}}(\lambda_0, q)}{4+2\log q}\quad\mbox{ and }\quad
C_{\text{\tiny{\ref{prop: main sec 2}}}}(\lambda_0, q)
:=\max(C_{\text{\tiny{\ref{prop: expectationFixedRank}}}}(\lambda_0, q), 1/c_{\text{\tiny{\ref{prop: main sec 2}}}}(\lambda_0, q)).
$$
Assume that $d$, $n$, $k$ satisfy the assumptions of the proposition.

Fix any $\ell\leq m$.
For each $r\leq \ell$, let $\Omega_{r,\ell}$ be defined as in Lemma~\ref{lem: mtuplesrank}.
For $\brk{v^1,\dots,v^\ell} \in \Omega_\ell$, we have $\E \prod_{i=1}^\ell W_{v^i}(d)=\rho_{d}^\ell$ 
whenever $v^1,\dots, v^\ell$ form a linearly 
independent set. Thus,
$$
     \E\sum_{\,v^1,\dots, v^\ell \in  \Omega_{\ell,\ell} }
        \prod_{i=1}^\ell W_{v^i}(d) =\E\sum_{\,v^1,\dots, v^\ell \in  \Omega_{\ell,\ell} }
        \prod_{i=1}^\ell \widetilde{W}_{v^i}(d)=\rho_{d}^\ell |\Omega_{\ell,\ell}|.
$$
Further, take any $r<\ell$, and observe that, in view of Corollary~\ref{cor: rank cor}, we have
$$
     \E\sum_{\,v^1,\dots, v^\ell \in  \Omega_{r,\ell} }
        \prod_{i=1}^\ell W_{v^i}(d)\leq C\rho_{d}^r \,\exp(-c_{\text{\tiny{\ref{prop: expectationFixedRank}}}}
(\lambda_0, q)d^4/n^3)\,|\Omega_{r,\ell}|.
$$
Applying the estimate from Lemma~\ref{lem: mtuplesrank} to the last expression, we obtain
\begin{align*}
      \E\sum_{\,v^1,\dots, v^\ell \in  \Omega_{r,\ell} }
      \prod_{i=1}^\ell W_{v^i}(d)&\leq C\rho_{d}^\ell\, |\Omega_{\ell,\ell}| \,
      \exp(-c_{\text{\tiny{\ref{prop: expectationFixedRank}}}}(\lambda_0, q)d^4/n^3)\;{ \ell \choose \ell- r} 
      \frac{\rho_d^{r-\ell}\brk{q^r}^{\ell-r}}{\prod_{i=r}^{\ell-1}\brk{q^k-q^i}}\\
      &\leq C\rho_{d}^\ell\, |\Omega_{\ell,\ell}| \,\exp(-c_{\text{\tiny{\ref{prop: expectationFixedRank}}}}(\lambda_0, q)d^4/n^3)\; 
      \frac{2^\ell q^{\ell^2}}{(\rho_d \,q^k)^{\ell-r}}.
\end{align*}
Further, using our assumptions on the parameters, we get
\begin{align*}
\frac{2^\ell q^{\ell^2}}{(\rho_d \,q^k)^{\ell-r}}
&\leq \exp\big(c_{\text{\tiny{\ref{prop: main sec 2}}}}(\lambda_0, q)\,(1+\log q)d^4/n^3\big)
\exp\bigg(\frac{c_{\text{\tiny{\ref{prop: main sec 2}}}}(\lambda_0, q)d^4(\ell-r)}{n^3 \ell}\bigg)\\
&\leq \exp\bigg(-r+\frac{1}{2}c_{\text{\tiny{\ref{prop: expectationFixedRank}}}}(\lambda_0, q)d^4/n^3\bigg).
\end{align*}
Thus,
$$
     \E\sum_{\,v^1,\dots, v^\ell \in  \Omega_{\ell} }
        \prod_{i=1}^\ell W_{v^i}(d) =\big(1+O(\exp(-c_{\text{\tiny{\ref{prop: main sec 2}}}}(\lambda_0, q)d^4/n^3))\big)
     \rho_{d}^\ell |\Omega_{\ell,\ell}|.
$$
On the other hand, applying Lemma~\ref{lem: mtuplesrank}, we get
$$
|\Omega_{\ell,\ell}|=\prod_{i=0}^{\ell-1}\brk{q^k-q^i}\geq \big(1-O(q^{\ell-k})\big)|\Omega_\ell|.
$$
Hence,
$$
\E\sum_{\,v^1,\dots, v^\ell \in  \Omega_{\ell} }
        \prod_{i=1}^\ell W_{v^i}(d)
=\big(1+O(\exp(-c_{\text{\tiny{\ref{prop: main sec 2}}}}(\lambda_0, q)d^4/n^3))+O(2^{-\frac{k}{2}})\big)
\E\sum_{\,v^1,\dots, v^\ell \in  \Omega_{\ell} }
        \prod_{i=1}^\ell \widetilde W_{v^i}(d).
$$
The result follows by applying formula \eqref{eq: -98529587-593}.
\end{proof}

\medskip

The rest of the section is devoted to proving Proposition~\ref{prop: expectationFixedRank}.

\begin{lem}
  Suppose $M  \in GL_{k}(\F)$ is a fixed invertible $k\times k$ matrix over the field $\F$.
  Let $Y$ be a random vector uniformly distributed in $\F^k$. Then, the image $MY$ is uniformly distributed in $\F^k$.
\end{lem}
As a corollary, we obtain
\begin{cor}\label{cor: lin indep}
  Suppose $a^1,\dots, a^\ell\in \F^k\setminus\{0\}$ are linearly independent fixed vectors, and let, as before, $X_1,\dots,X_k$
  be i.i.d.\ random vectors uniformly distributed in $\F^n$. Then the random vectors
$$
\sum_{i=1}^k a^{\ell'}_iX_i,\quad \ell'=1,\dots,\ell,$$
are mutually independent and uniformly distributed over $\F^n$. 
\end{cor}
\begin{proof}
  It is sufficient to show that for each $t\in[n]$, the random variables
  $\sum_{i=1}^k a^{\ell'}_iX_i(t)$, $\ell'=1,\dots,\ell$,
  are mutually independent and uniform over $\F$, where by $X_i(t)$ we denote the $t$-th
  component of $X_i$.

  If $\ell<k$ then we can find vectors $a^{\ell+1},a^{\ell+2},\dots, a^{k}$ such that 
  $a^1,\dots, a^{k}$ are linearly independent. Let 
  $M$ be the $k\times k$ matrix with rows $a^1,\dots a^{k}$ and let $Y$
  be the random vector in $\F^k$ with components 
  $X_1(t),\dots X_k(t)$. Applying the above lemma, we get that the components of the vector $MY$ are mutually independent
  and uniform over $\F$. The result follows.
\end{proof}

\begin{lem}
  \label{prop: wtof2Sum}
    For any $\lambda_0 \in (0,1)$ there is $c_{\text{\tiny{\ref{prop: wtof2Sum}}}}(\lambda_0, q)>0$
  with the following property.
  Suppose $d,n\in \mathbb{N}$ satisfy 
  $\frac{d}{n} \leq \lambda_0 (1-\frac{1}{q})$. 
  Let $\varepsilon=\min\set{\frac{1}{4},\,\frac{1-\lambda_0}{2}}$ and take $\kappa,\gamma>0$ such that
  $\gamma d$ and $\kappa d$ are integers,
  $(1-\varepsilon)\le \kappa \le \frac{n}{d}$, and $(1-\varepsilon)\le \gamma \le 1$.
  Let $V = (\underbrace{1,\dots,1}_{ \kappa d\,\, {\rm terms} },0,\dots, 0)\in\F^n$, and
  let $Y$ be a random vector uniformly distributed in $\F^n$. 
  Then, 
  $$
    \PP\set{ \wt{V+Y}\le d\,|\, \wt{Y}=\gamma d}\le 3\exp(-c_{\text{\tiny{\ref{prop: wtof2Sum}}}}(\lambda_0, q)d^4/n^3).
  $$
\end{lem}
\begin{proof}
  Let $J$ be a random subset of size $\gamma d$ uniformly distributed in 
  $[n]$. Let $u_1,\dots, u_{n}$ 
  be i.i.d random variables uniformly distributed in $\F\setminus\{0\}$, independent from $J$.
  Then the conditional distribution of $Y$ given $\wt{Y}=\gamma d$ coincides with the distribution of
  $$U=\brk{ u_1 {\bf 1}_{J}\brk{1},\dots, u_n {\bf 1}_J\brk{n}}^\top,$$
  where ${\bf 1}_J(\cdot)$ is the indicator function of $J$. Thus, 
  \begin{align*}
      \PP\set{ \wt{V+Y}\le d\,|\, \wt{Y}=\gamma d}
      = \PP\set{ \wt{V+U}\leq d}.
  \end{align*}
  Notice that 
  \begin{align}
    \label{eq: weightOfSum}
    \wt{V+U} = \kappa d + \gamma d - \abs{[\kappa d]\cap J} -
    \abs{\set{i\in [\kappa d]\cap J\,:\, u_i= -1}}.
  \end{align}
  The random variable $\abs{[\kappa d]\cap J}$ has the hypergeometric distribution given by
  $$
  \PP\{\abs{[\kappa d]\cap J}=\ell\}={n\choose \gamma d}^{-1}{n-\kappa d \choose \gamma d-\ell}{\kappa d \choose \ell},\quad \ell\geq 0.
  $$
  The expected size of $[\kappa d] \cap J$ is
  $\frac{\kappa d}{n}\gamma d \leq \kappa \gamma \lambda_0 (\frac{q-1}{q})d$. 
  Applying the large deviation inequality for the hypergeometric distribution (see, e.g. Talagrand \cite{Tal}), we then get
  \begin{align}\label{eq: aux 20985209572}
    \PP \set{
      \abs{\abs{[\kappa d]\cap J} -\frac{\kappa d}{n}\gamma d}
      \ge  t\gamma d}
    \le 2\exp(-2t^2\gamma d),\quad t>0.
  \end{align}
  Let $\Event_1(s)$ be the event that  
  $\abs{\abs{[\kappa d]\cap J} -\frac{\kappa d}{n}\gamma d}
  \le  \kappa s \brk{\frac{q-1}{q}}\gamma d$,
  where $s>0$ will be determined later.
  Further, let $\Event_2(s)$ be the 
  event that $$\abs{\set{i\in [\kappa d]\cap J\,:\, u_i=-1}}
  \le \brk{\frac{1}{q-1}+s}\abs{[\kappa d]\cap J}.$$
  By Hoeffding's inequality, we have 
  \begin{equation}\label{eq: aux 20985034985}
    \PP\big( \Event_2^c(s)\,\mid\, J\big)
    \le \exp\brk{-2s^2 \abs{[\kappa d]\cap J}}.
  \end{equation}
  Whenever both $\Event_1(s)$ and $\Event_2(s)$ hold, we have, in view of \eqref{eq: weightOfSum},
  \begin{align*}
    \wt{V+U} &\ge  \kappa d+\gamma d- 
    \Big(1+\frac{1}{q-1}+s\Big)\abs{[\kappa d] \cap J}\\
    &\ge \kappa d+\gamma d - 
     \Big(\frac{q}{q-1}+s\Big)\kappa\gamma d \Big(\frac{d}{n}+s\frac{q-1}{q}\Big).
  \end{align*}
  Now, we set $s:=\min\big(\frac{qd}{2n(q-1)},\frac{1-\lambda_0}{6}\big)$, so that the above relation implies
  \begin{align*}
    \wt{V+U} \ge d\Big(\kappa + \gamma -\kappa \frac{\lambda_0+1}{2}\gamma\Big)\quad
    \mbox{everywhere on $\Event_1(s)\cap \Event_2(s)$}.
  \end{align*}

  If $\kappa\frac{\lambda_0+1}{2}\ge 1$ then
  \begin{align*}
    \kappa + \gamma -\kappa \frac{\lambda_0+1}{2}\gamma 
    \ge \kappa + 1 -\kappa\frac{\lambda_0+1}{2}
    = 1 + \kappa \frac{1-\lambda_0}{2},
  \end{align*}
  which implies that $\wt{U+V}>d$ everywhere on $\Event_1(s)\cap \Event_2(s)$.

  Next, if
  $\kappa\frac{\lambda_0+1}{2}< 1$ then we have 
  \begin{align*}
    \kappa + \gamma -\kappa \frac{\lambda_0+1}{2}\gamma 
    \ge& (1-\varepsilon)+(1-\varepsilon) - (1-\varepsilon)^2\frac{\lambda_0+1}{2}\\
    =& 1 + \frac{1-\lambda_0}{2} -2\varepsilon \,\frac{1-\lambda_0}{2}-\varepsilon^2\frac{\lambda_0+1}{2}.
  \end{align*}
  Using the definition $\varepsilon =\min\big(\frac{1}{4},\,\frac{1-\lambda_0}{2}\big)$,
  we get
  \begin{align*}
    \kappa + \gamma -\kappa \frac{\lambda_0+1}{2}\gamma 
    \ge 1+ (1-3\varepsilon)\frac{1-\lambda_0}{2}>1
  \end{align*}
  which, again, implies that $\wt{U+V}>d$ everywhere on $\Event_1(s)\cap \Event_2(s)$. 

  It remains to estimate the probability of $\Event_1(s)\cap \Event_2(s)$. In view of \eqref{eq: aux 20985034985},
  the definition of $\Event_1(s)$, and the condition $s\leq \frac{qd}{2n(q-1)}$,
  we have
  $$\PP\big( \Event_2^c(s)\, |\, \Event_1(s)\big) \le 
  \exp(-c(\lambda_0, q)d^4/n^3),$$
  whereas \eqref{eq: aux 20985209572} yields
  $$\PP(\Event_1^c(s))\leq 2\exp(-c(\lambda_0, q) d^3/n^2)$$
  for some $c(\lambda_0,q)>0$ depending only on $\lambda_0$ and $p$.
  Hence, 
  $$
    \PP\set{\wt{V+Y}\le d\,|\, \wt{Y}=\gamma d} \le 
    \PP\set{ \brk{\Event_1(s)\cup \Event_2(s)}^c} \le 3\exp(-c(\lambda_0, q)d^4/n^3),
  $$
  and the result follows.
\end{proof}

\begin{proof}[Proof of Proposition~\ref{prop: expectationFixedRank}]
Notice that 
$$
  W_{v^t}(d) = {\bf 1}_{\left\{ \wt{\brk{\sum_{i=1}^k v^t_i X_i}}\le d\right\}}
    = {\bf 1}_{\left\{\wt{\brk{\sum_{i=1}^k c_tv^t_i X_i}}\le d\right\}}.
$$
Further, since the vectors $c_1v^1,\dots, c_sv^s$ are linearly independent, by 
Corollary \ref{cor: lin indep} the joint distribution of the vectors
$$\sum_{i=1}^k c_1v^1_iX_i,\,\sum_{i=1}^k c_2v^2_iX_i, \dots, \sum_{i=1}^k c_sv^s_iX_i, \;
\sum_{i=1}^k \brk{\sum_{t=1}^s c_tv^t_i}X_i$$ is the 
same as that of $Y_1,\dots, Y_s, \sum_{t=1}^s Y_t$, where $Y_1,\dots,Y_s$ are i.i.d.\ copies of $X_1$. Thus, 
\begin{align*}
  \E&\brk{\prod_{i=1}^s W_{v^i}(d)} W_{\sum_{i=1}^s c_iv^i}(d)\\
  &= \PP\big\{ \forall i\in [s]\,  {\rm wt}\brk{Y_i}\le d,\,{\rm and}\, {\rm wt}\brk{S+Y_s}\le d\big\}\\
  &= \E\big({\bf 1}_{\{\wt{Y_i}\leq d\,\forall\,i\leq s-2\}}
\PP\big\{\wt{Y_{s-1}}\le d,\,\wt{Y_{s}}\le d,\, \wt{S+Y_s}\le d\;\mid\; Y_1,\dots,Y_{s-2}\big\}\big),
\end{align*}
where $S := \sum_{i=1}^{s-1} Y_i$. 
Set $P_1:=\PP\set{\wt{Y_{s-1}}\le d,\, \wt{Y_{s}}\le d,\,\wt{S+Y_s}\le d\,|\, Y_1,\dots,Y_{s-2}}$.
We shall break $P_1$ into two components:
\begin{align*}
  P_1 =&
  \PP\set{\wt{Y_{s-1}}\le d,\, \wt{Y_{s}}\le d,
  \,\wt{S+Y_s}\le d,\,\wt{S}< (1-\varepsilon)d\,|\, Y_1,\dots,Y_{s-2}}\\
  &+\PP\set{\wt{Y_{s-1}}\le d,\,\wt{Y_{s}}\le d,
  \,\wt{S+Y_s}\le d,\,\wt{S}\ge (1-\varepsilon)d\,|\, Y_1,\dots,Y_{s-2}}\\
  \le &
  \underbrace{\PP\set{\wt{Y_{s}}\le d,
  \,\wt{S}< (1-\varepsilon)d\,|\, Y_1,\dots,Y_{s-2}}}_{A}\\
  &+\underbrace{\PP\set{\wt{Y_{s-1}}\le d,\,\wt{Y_{s}}\le d,
  \,\wt{S+Y_s}\le d,\,\wt{S}\ge (1-\varepsilon)d\,|\, Y_1,\dots,Y_{s-2}}}_B.
\end{align*}
Notice that $Y_1,\dots,Y_{s-2},S$ and $Y_s$
are mutually independent (and uniform on $\F^n$), so we have 
$$
A= \PP\set{\wt{Y_{1}}\le d,\wt{Y_{2}}< (1-\varepsilon)d}
\le \rho_{d}^2 \exp(-c(\lambda_0, q)d)
$$
for some $c(\lambda_0, q)>0$,
where we have used the second assertion of Proposition~\ref{prop: 09860986}.
For the second summand, 
\begin{align*}
  B&=\E\big({\bf 1}_{\{\wt{Y_{s-1}}\le d\}}\,{\bf 1}_{\{\wt{Y_{s}}\le d\}}\,
  {\bf 1}_{\{\wt{S+Y_s}\le d\}}\,{\bf 1}_{\{\wt{S}\ge (1-\varepsilon)d\}}\;\mid\;Y_1,\dots,Y_{s-2}\big)\\
&=\E\big(\E\big({\bf 1}_{\{\wt{Y_{s}}\le d\}}\,{\bf 1}_{\{\wt{S+Y_s}\le d\}}
\;\mid\;Y_1,\dots,Y_{s-2},\,{\bf 1}_{\{\wt{Y_{s-1}}\le d\}},\,{\bf 1}_{\{\wt{S}\ge (1-\varepsilon)d\}},\,\wt{S}
\big)\cdot\\
&\hspace{0.85cm}{\bf 1}_{\{\wt{Y_{s-1}}\le d\}}\,{\bf 1}_{\{\wt{S}\ge (1-\varepsilon)d\}}\;\mid\;Y_1,\dots,Y_{s-2}\big)\\
&=\E\big(\E\big({\bf 1}_{\{\wt{Y_{s}}\le d\}}\,{\bf 1}_{\{\wt{S+Y_s}\le d\}}
\;\mid \;\wt{S}
\big)\cdot
{\bf 1}_{\{\wt{Y_{s-1}}\le d\}}\,{\bf 1}_{\{\wt{S}\ge (1-\varepsilon)d\}}\;\mid\;Y_1,\dots,Y_{s-2}\big),
\end{align*}
where we used the fact that the conditional expectation
$$
\E\big({\bf 1}_{\{\wt{Y_{s}}\le d\}}\,{\bf 1}_{\{\wt{S+Y_s}\le d\}}
\;\mid\;Y_1,\dots,Y_{s-2},\,{\bf 1}_{\{\wt{Y_{s-1}}\le d\}},\,{\bf 1}_{\{\wt{S}\ge (1-\varepsilon)d\}},\,\wt{S}
\big)
$$
is actually measurable with respect to $\wt{S}$.
On the event $\{\wt{Y_{s-1}}\le d,\;\wt{S}\ge (1-\varepsilon)d\}$ we have
\begin{align*}
\E\big(&{\bf 1}_{\{\wt{Y_{s}}\le d\}}\,{\bf 1}_{\{\wt{S+Y_s}\le d\}}
\;\mid \;\wt{S}
\big)\\
    =&\PP\set{\wt{Y_{s}}\le (1-\varepsilon)d,
    \,\wt{S+Y_s}\le d\,|\, \wt{S}}\\
    &+\PP\set{(1-\varepsilon)d < \wt{Y_{s}}\le d,
    \,\wt{S+Y_s}\le d\,|\, \wt{S}}\\
    \le &
    \PP\set{\wt{Y_{s}}\le (1-\varepsilon)d}
    +\PP\set{(1-\varepsilon)d < \wt{Y_{s}}\le d,
    \,\wt{S+Y_s}\le d\,|\, \wt{S}}\\
    \le & \rho_{d}\exp(-c(\lambda_0, q)d)
    +\E\big({\bf 1}_{\{(1-\varepsilon)d < \wt{Y_{s}}\le d\}}
    \PP\{\wt{S+Y_s}\le d\,|\, \wt{S},\,\wt{Y_s}\}\,|\, \wt{S}\big),
\end{align*}
where we again used the second assertion of Proposition~\ref{prop: 09860986}.
To estimate the second summand in the last expression, we apply Lemma~\ref{prop: wtof2Sum}.
First, let us make the following observation.

For any fixed $n \times n$ diagonal matrix $D$ whose diagonal entries are non-zero elements in $\mathbb{F}$
and any $v\in \mathbb{F}^n$, we have $\wt{v}=\wt{Dv}$, and, furthermore,
$DX_1$ and $X_1$ have the same distribution. 
Similarly, for any fixed $n\times n$ permutation matrix $P$, $\wt{v}=\wt{Pv}$ for any vector $v\in \mathbb{F}^n$, and $PX_1$ and $X_1$ 
have the same distribution. 
Now, note that we can construct a {\it random} permutation matrix $P$ and a random diagonal matrix $D$
with non-zero diagonal elements,
both measurable with respect to $S$, such that everywhere on the probability space
$
  PDS = V,
$
where $V = \big(\underbrace{1,\dots,1}_{\wt{S}\; {\rm terms}}
,0,\dots,0\big)^\top$.
Then we have
\begin{align*}
&\PP\{\wt{S+Y_s}\le d\,|\,\wt{S},\,\wt{Y_s}\}\\
&\hspace{1cm}= \PP\{\wt{V+PDY_s}\le d\,|\,\wt{S},\,\wt{Y_s}\}\\
&\hspace{1cm}= \PP\{\wt{V+Y_s}\le d\,|\,\wt{S},\,\wt{Y_s}\},
\end{align*}
where we used that $Y_s$ is independent from $P,D,S$.

The above observation allows us to use Lemma~\ref{prop: wtof2Sum}:
conditioning on any value $k_1\in[(1-\varepsilon)d,n]$ of $\wt{S}$ and any value $k_2\in[(1-\varepsilon)d,d]$
of $\wt{Y_s}$, we have
$$
\PP\{\wt{S+Y_s}\le d\,|\, \wt{S}=k_1,\,\wt{Y_s}=k_2\}\leq 3\exp(-c'(\lambda_0, q)d^4/n^3)
$$
for some $c'(\lambda_0, q)>0$ depending only on $\lambda_0$ and $q$.
Thus, everywhere on $\{\wt{Y_{s-1}}\le d,\,\wt{S}\ge (1-\varepsilon)d\}$ we have
$$
\E\big({\bf 1}_{\{\wt{Y_{s}}\le d\}}\,{\bf 1}_{\{\wt{S+Y_s}\le d\}}
\;\mid \;\wt{S}
\big)\leq 4\rho_d\exp(-c''(\lambda_0, q)d^4/n^3),
$$
where $c''(\lambda_0, q):=\min(c(\lambda_0, q),c'(\lambda_0, q))$.
That, in turn, implies
$$
B\leq 4\rho_{d}^2\exp(-c''(\lambda_0, q)d^4/n^3).
$$
Then, we have 
$$
    P_1 \le 5\rho_{d}^2 \exp(-c''(\lambda_0,q)d^4/n^3)
$$
and the result follows since
$\PP\set{\wt{Y_i}\le d,\, \forall i\leq s-2} 
= \rho_{d}^{s-2}$.
\end{proof}

\section{Analysis of the distribution of the minimal distance}
\label{sec: analysis of the distribution of the minimal distance}

The goal of this section is to prove our main result comparing the distributions of the minimal distance of the random linear code,
with the minimum $w_{\min}$ of the weights of the random vectors $Y_a$, $a\in\F^k\setminus\{0\}$
(defined earlier in the paper).

First, we state the ``technical'' version of the result:
\begin{thm}\label{prop: main prop}
For any $\lambda_0 \in (0,1)$ there are $c_{\text{\tiny{\ref{prop: main prop}}}}(\lambda_0, q)>0$ and
$C_{\text{\tiny{\ref{prop: main prop}}}}(\lambda_0, q)>0$ with the following property.
Let $n\geq 1$, and take any $L\geq e$.
Assume further that $k$ satisfies $C_{\text{\tiny{\ref{prop: main prop}}}}(\lambda_0, q)L\log L\leq k\leq n$,
and take any $d$ such that
$$C_{\text{\tiny{\ref{prop: main prop}}}}(\lambda_0, q)\sqrt{L}n^{3/4}\leq d\leq \lambda_0 \Big(1-\frac{1}{q}\Big)n,$$
and
$c_{\text{\tiny{\ref{prop: main prop}}}}(\lambda_0, q) L\geq q^k\rho_d\geq
\exp\big(-\frac{c_{\text{\tiny{\ref{prop: main prop}}}}(\lambda_0, q)d^2}{n^{3/2}}\big)$.
Let, as before, $X_1,\dots,X_k$ be i.i.d.\ random vectors uniformly distributed in $\F^n$, and denote
$$
d_{\min}:=\min\bigg\{\wt{\sum\nolimits_{i=1}^k a_i X_i},\;a\in\F^k\setminus\{0\}\bigg\}.
$$
Then
$$
\big|\PP\{d_{\min}\leq d\}-\PP\{w_{\min}\leq d\}\big|=O(\exp(-L)).
$$
\end{thm}
The theorem provides some freedom of the choice of the parameters, and includes a regime when
the ratio $k/n$ converges to one when $n\to\infty$. At the same time, we would like to provide a cleaner statement
for the most important regime when $k/n$ is ``separated'' from both $0$ and $1$. We will obtain Theorem \ref{main} as a corollary of Theorem \ref{prop: main prop}.

\bigskip

For each $r\geq 0$, we let
$$
M_d(r):=\PP\big\{Z_d=r\big\},\quad \widetilde M_d(r):=\PP\big\{\widetilde Z_d=r\big\},
$$
so that
\begin{align*}
\PP\big\{d_{\min}\leq d\big\}&=\PP\big\{Z_d>0\big\}=\sum\limits_{r=1}^\infty M_d(r);\\
\PP\big\{w_{\min}\leq d\big\}&=\PP\big\{Z_d>0\big\}=\sum\limits_{r=1}^\infty \widetilde M_d(r).
\end{align*}
Observe further that the numbers $M_d(r)$ and $\widetilde M_d(r)$ satisfy the relations
$$
\sum\limits_{r=1}^\infty M_d(r) r^m=\E Z_d^m,\quad
\sum\limits_{r=1}^\infty \widetilde M_d(r) r^m=\E \widetilde Z_d^m,\quad m\geq 1.
$$
These identities, together with the relations between $\E Z_d^m$ and $\E \widetilde Z_d^m$
obtained in the previous section, will allow us to compare $M_d(r)$ with $\widetilde M_d(r)$, hence
bound the distance between the distributions of $d_{\min}$ and $w_{\min}$.
Let us start by recording a moment growth estimate for $\widetilde{Z}_d$: 
\begin{lem}\label{l: moment growth}
  We have
  \begin{align*}
     \brk{\E \widetilde{Z}_d^\ell}^{1/\ell} \le C_{\text{\tiny{\ref{l: moment growth}}}} 
     \begin{cases} 
      \frac{q^k\rho_{d}}{q-1}, & \mbox{ if }\ell\le  \frac{q^k\rho_{d}}{q-1}, \\
       \frac{\ell}{\log(e \ell(q-1) / (q^k\rho_{d}))}, & \mbox{ if }\ell \ge \frac{q^k\rho_{d}}{q-1}.
    \end{cases}
  \end{align*}
Here, $C_{\text{\tiny{\ref{l: moment growth}}}}>0$ is a universal constant.
\end{lem}
\begin{proof}
  Notice that
  \begin{align*}
  \E \widetilde{Z}_d^\ell = \sum_{m=1}^\ell S(\ell,m)\brk{q-1}^{\ell-m} \abs{\Omega_m}\rho_{d}^m,
  \end{align*}
  where we use the notation from the previous section.
  The cardinality of $\Omega_m$ is 
  \begin{align*}
    |\Omega_m|&=\brk{q^k-1}\brk{q^k-(q-1)-1}\brk{q^k-2(q-1)-1}\cdots 
    \brk{q^k-(m-1)(q-1)-1} \\
&\le q^{km}.
  \end{align*}
  Thus, 
  \begin{align*}
    \sum_{m=1}^\ell S(\ell,m)\brk{q-1}^{\ell-m} \abs{\Omega_m}\rho_{d}^m
    \le  (q-1)^{\ell}\sum_{m=1}^{\ell} S(\ell,m)\brk{ \frac{q^k\rho_{d}}{q-1} }^m.
  \end{align*}

  Let $\lambda:=  \frac{q^k\rho_{d}}{q-1} $. We will use an upper
  estimate for $S(\ell,m)$ from \cite[Theorem 3]{RD69}:
  \begin{align*}
    S(\ell,m) \le \frac{1}{2}{\ell \choose m} m^{\ell-m}\le 2^{\ell+1} m^{\ell-m}.
  \end{align*}
  With this bound and the substitution $m=t\ell$ we have 
  \begin{align*}
    \sum_{m=1}^\ell S(\ell,m)\lambda^{m} \le \ell 2^{\ell+1} \brk{\max_{ t\in [0,1]} 
    \brk{\brk{t\ell}^{(1-t)}\lambda^t}}^\ell.
  \end{align*} 

  When $\ell \le \lambda$, we get 
  \begin{align*}
    \max_{ t\in [0,1]} 
    \brk{\brk{t\ell}^{(1-t)}\lambda^t} \le \lambda 
  \end{align*}
  which finishes the proof of the first inequality.
  Now we assume $\ell\ge \lambda$. We will use 
  a standard argument from calculus. 
  Consider the derivative
  \begin{align*}
    \frac{d}{dt}\brk{t\ell}^{(1-t)}\lambda^t&=
  \frac{d}{dt}\exp(\log(t\ell)(1-t)+\log(\lambda)t)\\
    &=\brk{
      \frac{1}{t}(1-t)-\log(t\ell)+\log(\lambda)
    } 
    \brk{t\ell}^{(1-t)}\lambda^t\\
    &=\brk{\frac{1}{t}-\log(t)-\brk{1+\log\Big(\frac{\ell}{\lambda}\Big)}}
    \brk{t\ell}^{(1-t)}\lambda^t.
  \end{align*}
  Notice that $\frac{1}{t}-\log(t)-\brk{1+\log(\frac{\ell}{\lambda})}$ is a monotone decreasing function which takes value 
  $\infty$ at $t=0$ and $-\log(\frac{\ell}{\lambda})<0$ at $t=1$. 
  Thus, the maximum of $\brk{t\ell}^{(1-t)}\lambda^t$
  is achieved when
  \begin{align*}
    \frac{1}{t}-\log(t) =  \brk{1+\log\Big(\frac{\ell}{\lambda}\Big)}.
  \end{align*}
  Now we fix $t\in (0,1)$ to be the constant satisfying the above equation.
  Since $  \frac{1}{t} \ge -\log(t) \ge 0$ on $t\in [0,1]$, we have 
\begin{align*}
  \frac{1}{\brk{1+\log(\frac{\ell}{\lambda})}} \le  t  \le  \frac{2}{\brk{1+\log(\frac{\ell}{\lambda})}}.
\end{align*}

Furthermore, we have 
$\lambda \le \frac{\ell}{\brk{1+\log(\frac{\ell}{\lambda})}}$ since 
 $\frac{x}{1+\log(x)}\ge 1$ for $x\ge 1$. We conclude that 
$$
  \max_{t\in[0,1]}\brk{t\ell}^{(1-t)}\lambda^t \le \frac{2\ell}{\brk{1+\log(\frac{\ell}{\lambda})}}
$$
and the statement of the lemma follows.

\end{proof}

Next, fix an integer parameter $h\geq 1$ (its value will be defined later), and define the $h\times h$ square matrix
$B=(b_{ij})$ as 
$$
b_{ij}=j^i,\quad i,j=1,\dots,h.
$$
The next lemma can be easily checked by a straightforward computation.
\begin{lem}
Let $B=(b_{ij})$ be as above.
Then $B$ is invertible, and the entries of the inverse matrix $B^{-1}=(b_{ij}')$ are given by
$$
b_{ij}'=\begin{cases}
\frac{(-1)^{j-1}\sum\limits_{\substack{1\leq m_1<\dots< m_{h-j}\leq h,\\m_1,\dots,m_{h-j}\neq i}}m_1\dots m_{h-j}}
{i\prod\limits_{1\leq m\leq h,\,m\neq i}(m-i)},&\mbox{if }j<h;\\
\frac{1}{i\prod\limits_{1\leq m\leq h,\,m\neq i}(i-m)},&\mbox{if }j=h.
\end{cases}
$$
\end{lem}
In what follows, we will not need a precise formula for the entries of the inverse; just a crude upper bound will be sufficient:
\begin{cor}\label{cor: crude bound for inv}
With the above notation, we have
$$
|b_{ij}'|\leq \frac{{h\choose j}h^{h-j}}{((\lfloor h/2\rfloor-1)!)^2}\leq C_{\text{\tiny{\ref{cor: crude bound for inv}}}}^h h^{-j},
$$
where $C_{\text{\tiny{\ref{cor: crude bound for inv}}}}>0$ is a universal constant.
\end{cor}

Denote the vector $(M_d(1),\dots,M_d(h))^\top$ by $V$, and the vector
$(\widetilde M_d(1),\dots,\widetilde M_d(h))^\top$ by $\widetilde V$.
Further, let $U:=(\E Z_d,\dots,\E Z_d^h)^\top$, and $\widetilde U:=(\E \widetilde Z_d,\dots,\E \widetilde Z_d^h)^\top$,
and, finally, define the ``error vectors''
$$
E:=\Big(\sum_{r=h+1}^\infty r^i M_d(r)\Big)_{i=1}^h,\quad \widetilde E:=\Big(\sum_{r=h+1}^\infty r^i \widetilde M_d(r)\Big)_{i=1}^h.
$$
In view of the above,
$$
B V+E=U,\quad B\widetilde V+\widetilde E=\widetilde U,
$$
whence the difference $V-\widetilde V$ can be expressed as
$$
V-\widetilde V=B^{-1}(U-\widetilde U)-B^{-1}(E-\widetilde E).
$$
Let us first estimate $B^{-1}(U-\widetilde U)$:
\begin{lem}\label{lem: moment err}
    Suppose $d,n\in \mathbb{N}$ satisfy
    $\frac{d}{n} \leq \lambda_0 (1-\frac{1}{q})$, and $d^{2}/n^{3/2}\geq
C_{\text{\tiny{\ref{prop: main sec 2}}}}(\lambda_0, q)$.
Assume additionally that $h\geq q^k\rho_d\geq \exp\big
(-\frac{c_{\text{\tiny{\ref{prop: main sec 2}}}}(\lambda_0, q)d^4}{n^3 h}\big)$,
$h\log_2 C_{\text{\tiny{\ref{cor: crude bound for inv}}}}+h\log_2 C_{\text{\tiny{\ref{l: moment growth}}}}
+h+h\log(hp-h)\leq k/4$
and $h\le \frac{c_{\text{\tiny{\ref{prop: main sec 2}}}}(\lambda_0, q)}
{\log_2 C_{\text{\tiny{\ref{cor: crude bound for inv}}}}+
\log_2 C_{\text{\tiny{\ref{l: moment growth}}}}+2+\log(q-1)}\frac{d^2}{n^{3/2}}$.
Then the absolute value of each component of the vector $B^{-1}(U-\widetilde U)$
is bounded above by
$$
O\bigg(\exp\Big(-\frac{1}{2}c_{\text{\tiny{\ref{prop: main sec 2}}}}(\lambda_0, q)d^4/n^3\Big)+2^{-k/4}\bigg).
$$
\end{lem}
\begin{proof}
Fix any $i,j\leq h$. Applying Proposition~\ref{prop: main sec 2} in combination with Lemma~\ref{l: moment growth}, we get
$$
|\E Z_d^j-\E\widetilde Z_d^j|\leq O\big(\exp(-c_{\text{\tiny{\ref{prop: main sec 2}}}}(\lambda_0, q)d^4/n^3)+2^{-k/2}\big)\,
C_{\text{\tiny{\ref{l: moment growth}}}}^j
     \begin{cases} 
      (q^{k}\rho_{d})^j, & \mbox{ if }j\le e \frac{q^k\rho_{d}}{q-1}, \\
       \Big(\frac{j(q-1)}{\log(j(q-1) / (q^k\rho_{d}))}\Big)^j, & \mbox{ if }j > e\frac{q^k\rho_{d}}{q-1}.
    \end{cases}
$$
On the other hand, according to Corollary~\ref{cor: crude bound for inv}, we have
$$
|b_{ij}'|\leq C_{\text{\tiny{\ref{cor: crude bound for inv}}}}^h.
$$
Hence, the $i$-th component of $B^{-1}(U-\widetilde U)$ can be bounded from above by
$$
O\Big(C_{\text{\tiny{\ref{cor: crude bound for inv}}}}^h\,h\,
\big(\exp(-c_{\text{\tiny{\ref{prop: main sec 2}}}}(\lambda_0, q)d^4/n^3)+2^{-k/2}\big)\,
C_{\text{\tiny{\ref{l: moment growth}}}}^h(h(q-1))^h\Big).
$$
Using the assumptions on parameters, we get the result.
\end{proof}

By a slightly more careful argument, we get an estimate on the term $B^{-1}(E-\widetilde E)$:
\begin{lem}\label{lem: error}
    Suppose $d,n\in \mathbb{N}$ satisfy
    $\frac{d}{n} \leq \lambda_0 (1-\frac{1}{q})$, and $d^2/n^{3/2}\geq C_{\text{\tiny{\ref{prop: main sec 2}}}}(\lambda_0, q)$.
Assume additionally that
$$e^{-8C_{\text{\tiny{\ref{l: moment growth}}}}C_{\text{\tiny{\ref{cor: crude bound for inv}}}}(q-1)}h\geq q^k\rho_d\geq \exp\bigg(
-\frac{c_{\text{\tiny{\ref{prop: main sec 2}}}}(\lambda_0, q)d^4}{4n^3 h}\bigg),$$
and $h\le \frac{c_{\text{\tiny{\ref{prop: main sec 2}}}}(\lambda_0, q)}{4}\frac{d^2}{n^{3/2}}$.
Then
$$
\sum_{r=h+1}^\infty M_d(r),\sum_{r=h+1}^\infty \widetilde M_d(r)=O(2^{-h}),
$$
and the absolute value of each component of the vector $B^{-1}(E-\widetilde E)$
is bounded above by
$
O(2^{-h})
$.
\end{lem}
\begin{proof}

  Set $m:=4h$, and observe that $m(q-1)/e\geq q^k\rho_d\geq
\exp\big(-\frac{c_{\text{\tiny{\ref{prop: main sec 2}}}}(\lambda_0, q)d^4}{n^3 m}\big)$
and, furthermore, $m\le c_{\text{\tiny{\ref{prop: main sec 2}}}}(\lambda_0, q)d^2/n^{3/2}$.
By applying the comparison Proposition~\ref{prop: main sec 2}, the bound given by Lemma~\ref{l: moment growth}
and Markov's inequality, we get for any $r>h$:
$$
M_d(r),\widetilde M_d(r)=O\bigg( C_{\text{\tiny{\ref{l: moment growth}}}}^m  r^{-m}
       \Big(\frac{m(q-1)}{\log(m(q-1) / (q^k\rho_{d}))}\Big)^{m}\bigg).
$$
Further, our conditions on $h$ and $q^k \rho_d$ imply that
$$
\log(m(q-1) / (q^k\rho_{d}))\geq 8C_{\text{\tiny{\ref{l: moment growth}}}}C_{\text{\tiny{\ref{cor: crude bound for inv}}}} (q-1)
=2C_{\text{\tiny{\ref{l: moment growth}}}}C_{\text{\tiny{\ref{cor: crude bound for inv}}}} m(q-1)/h.
$$
Thus,
$$
M_d(r),\widetilde M_d(r)=O\big( h^m  r^{-m}(2C_{\text{\tiny{\ref{cor: crude bound for inv}}}})^{-m}
\big),
$$
and we get the first assertion of the lemma after summing up.
Further, taking into account the definition of the vectors $E$ and $\widetilde E$, we obtain for every $j\leq h$:
$$
E_j,\widetilde E_j=O\bigg( \sum_{r=h+1}^\infty h^m  r^{j-m}(2C_{\text{\tiny{\ref{cor: crude bound for inv}}}})^{-m}\bigg)
=O\big(h^{j} (2C_{\text{\tiny{\ref{cor: crude bound for inv}}}})^{-m}\big).
$$
At the same time, by Corollary~\ref{cor: crude bound for inv} we have
$|b_{ij}'|\leq C_{\text{\tiny{\ref{cor: crude bound for inv}}}}^h h^{-j}$,
so that the $i$-th component of the vector $B^{-1}E$ can be bounded above by
$O(2^{-h})$. The same is true for $B^{-1}\widetilde E$, and
the result follows.
\end{proof}

\medskip

\begin{proof}[Proof of Theorem~\ref{prop: main prop}]
We start by defining the constants.
Let
\begin{align*}
C_{\text{\tiny{\ref{prop: main prop}}}}(\lambda_0, q):=\max\bigg(
&64\big(2+\log(16q-16)+\log_2 C_{\text{\tiny{\ref{cor: crude bound for inv}}}}+\log_2 C_{\text{\tiny{\ref{l: moment growth}}}}\big),
\sqrt{C_{\text{\tiny{\ref{prop: main sec 2}}}}(\lambda_0, q)},\\
&4\Big(\frac{c_{\text{\tiny{\ref{prop: main sec 2}}}}(\lambda_0, q)}
{\log_2 C_{\text{\tiny{\ref{cor: crude bound for inv}}}}+
\log_2 C_{\text{\tiny{\ref{l: moment growth}}}}+4+\log(q-1)}\Big)^{-1/2}\bigg)
\end{align*}
and
$$
c_{\text{\tiny{\ref{prop: main prop}}}}(\lambda_0, q):=\min\big(e^{-8C_{\text{\tiny{\ref{l: moment growth}}}}C_{\text{\tiny{\ref{cor: crude bound for inv}}}}(q-1)},c_{\text{\tiny{\ref{prop: main sec 2}}}}(\lambda_0, q)/4\big),
$$
and fix any $L$, $d$, $k$ satisfying the conditions of the theorem.
We set $h:=\lfloor 16L\rfloor$.
Observe that 
$$
h\log_2 (C_{\text{\tiny{\ref{cor: crude bound for inv}}}}C_{\text{\tiny{\ref{l: moment growth}}}})
+h+h\log(hq-h)\leq k/4
,\;
h\le \frac{c_{\text{\tiny{\ref{prop: main sec 2}}}}(\lambda_0, q)}
{\log_2 C_{\text{\tiny{\ref{cor: crude bound for inv}}}}+
\log_2 C_{\text{\tiny{\ref{l: moment growth}}}}+4+\log(q-1)}\frac{d^2}{n^{3/2}},
$$
and that the product $q^k \rho_d$ satisfies
$$
e^{-8C_{\text{\tiny{\ref{l: moment growth}}}}C_{\text{\tiny{\ref{cor: crude bound for inv}}}}(q-1)}h\geq q^k\rho_d\geq \exp\bigg(
-\frac{c_{\text{\tiny{\ref{prop: main sec 2}}}}(\lambda_0, q)d^4}{4n^3 h}\bigg).
$$
A combination of Lemmas~\ref{lem: moment err} and~\ref{lem: error} then gives
\begin{align*}
\big|&\PP\{Z_d>0\}-\PP\{\widetilde Z_d>0\}\big|
\leq \|V-\widetilde V\|_1+\sum_{r=h+1}^\infty M_d(r)+\sum_{r=h+1}^\infty \widetilde M_d(r)\\
&=O(2^{-h})+O(h\,2^{-h})+O\bigg(h\exp\Big(-\frac{1}{2}c_{\text{\tiny{\ref{prop: main sec 2}}}}(\lambda_0, q)d^4/n^3\Big)+h\,2^{-h/4}\bigg).
\end{align*}
Our definition of $h$ then implies that
$$
\big|\PP\{Z_d>0\}-\PP\{\widetilde Z_d>0\}\big|=O(2^{-h/8}),
$$
and the result follows.
\end{proof}

\begin{lem} \label{prop:uniformWhenLinearIndependent}
  Let $H_1,H_2$ be two $k$-dimensional subspaces of $\F^n$.
  Then, 
  \begin{align*}
    \PP\set{H_1= \rm{span}(X_1,\dots,X_k)}=    
    \PP\set{H_2= \rm{span}(X_1,\dots,X_k)}
  \end{align*}
  where $X_1,\dots,X_k$ are i.i.d.\ random vectors uniformly 
  distributed in $\F^n$.
  As a consequence, conditioned on the event that 
  $X_1,\dots,X_k$ are linearly independent, 
  $\rm{span}(X_1,\dots,X_k)$ is a random $k$-dimensional subspace 
  uniformly distributed over all $k$-dimensional subspaces of $\F^n$.
\end{lem}
\begin{proof}
  Let $M$ be the $n\times k$ matrix with columns $X_1,\dots,X_k$. 
  Then, $M$ is uniformly chosen among all $n\times k$
  matrices with coefficients in $\F$.

  Let $H$ be a $k$-dimensional subspace. Then $\PP\set{ H = \rm{span}\brk{\mbox{column vectors of }\,M}}$
  is equal to the ratio of the number of $n\times k$ matrices whose column vectors span $H$, and the 
  number of all $n\times k$ matrices over $\F$.
  The number 
  of $n\times k$ matrices whose column vectors span $H$ is 
  $$
      \brk{q^k - 1}\brk{q^k-q}\brk{q^k-q^2}\cdots 
      \brk{q^k - q^{k-1}} 
  $$
    where $q^k -q^{i-1}$ represents the number of choices
    of the $i-$th column from $H$ which is not a linear combination of 
    first $i-1$ columns, with the assumption that the first $i-1$
    columns are independent. Since it does not depends on $H$, 
    it is the same for all $k$-dimensional subspaces. Thus, 
    the first statement follows. 

    Further, the collection
    of $n\times k$ matrices whose column vectors span $H$
    is a subset of the collection of $n\times k$ 
    matrices whose column vectors are linearly independent. Thus, 
    \begin{align*}
      &\PP\set{ H = \rm{span}\brk{\mbox{column vectors of }\,M}\,|\, 
      \mbox{column vectors of }M\mbox{ are linearly independent}}\\
      =& \frac{\mbox{The number 
      of $n\times k$ matrices whose column vectors span $H$}}
      {\mbox{The number of $n\times k$ matrices whose 
      column vectors are linearly independent}},
    \end{align*}    
    and the second assertion follows. 

\end{proof}

\begin{proof}[Proof of Theorem \ref{main}]
Fix any $R_1<R_2$ in the interval $(0,1)$. Without loss of generality, we can assume that $n$ is large.
First, we observe that there exist numbers
$\lambda_0=\lambda_0(R_1,R_2,q)\in(0,1)$, $\widetilde c=\widetilde c(R_1,R_2,q)>0$, and
$c'=c'(R_1,R_2,q)>0$ such that
$$
q^{R_1 n}\rho_{\lfloor \lambda_0 (n-n/q)\rfloor}\geq \exp(c' n),
$$
and
$$
q^{R_2 n}\rho_{\lceil \widetilde c n\rceil}\leq \exp(-c'n).
$$
These estimates can be obtained, in particular, with help of Proposition~\ref{prop: 09860986}.

Set $L:={\widetilde c}^2\sqrt{n}/C_{\text{\tiny{\ref{prop: main prop}}}}^2(\lambda_0, q)$.
Since $k\geq R_1 n$ and $n$ is large, we can assume
$$k\geq C_{\text{\tiny{\ref{prop: main prop}}}}(\lambda_0, q)L\log L.$$
Further, let $d_1$ be the smallest integer in the interval
$$\Big[C_{\text{\tiny{\ref{prop: main prop}}}}(\lambda_0, q)\sqrt{L}n^{3/4},\lambda_0 \Big(1-\frac{1}{q}\Big)n\Big]$$
such that
$q^k\rho_{d_1}\geq \exp\big(-\frac{c_{\text{\tiny{\ref{prop: main prop}}}}(\lambda_0, q)d_1^2}{n^{3/2}}\big)$,
and let $d_2$ be the largest integer in the same interval, such that
$$
q^k\rho_{d_2}\leq c_{\text{\tiny{\ref{prop: main prop}}}}(\lambda_0, q) L.
$$
Note that since $C_{\text{\tiny{\ref{prop: main prop}}}}(\lambda_0, q)\sqrt{L}n^{3/4}= \widetilde c n$,
the numbers $d_1$ and $d_2$ are well defined (for large enough $n$), and, moreover,
\begin{equation}\label{eq: aux 985204982}
q^k\rho_{d_1}\leq \exp(-c'' \sqrt{n}), \quad q^k\rho_{d_2}\geq c'' \sqrt{n}
\end{equation}
for some $c''(R_1,R_2,q)>0$.

Let $X_1,\dots,X_k$ be i.i.d.\ random vectors uniformly distributed in $\F^n$.
Denote by $\Event$ the event that the vectors are linearly independent.
By Lemma~\ref{prop:uniformWhenLinearIndependent},
conditioned on $\Event$, the linear span of $X_1,\dots,X_k$ is equidistributed with $\mathcal C$
. For any $d\in[d_1,d_2]$, applying Theorem~\ref{prop: main prop}, we get
$$
\big|\PP\{d_{\min}\leq d\}-\PP\{w_{\min}\leq d\}\big|=O(\exp(-L)),
$$
whence
$$
\big|\PP\{d_{\min}\leq d\,|\,\Event\}-\PP\{w_{\min}\leq d\}\big|=O(\exp(-L))+\PP(\Event^c).
$$
On the other hand, it is not hard to check that $\PP(\Event^c)=O(\exp(-\hat c n))$
for some $\hat c=\hat c(R_2)>0$. Thus, we obtain the required estimate for the difference $|F_{\rm dmin}(x)-F_{\rm wmin}(x)|$
within the interval $x\in[d_1,d_2]$. To complete the proof, it remains to notice that, in view of \eqref{eq: aux 985204982},
we have
$$
F_{\rm wmin}(d_1)\leq q^k\rho_{d_1}\leq \exp(-c'' \sqrt{n}),
$$ 
and
$$
1-F_{\rm wmin}(d_2)\leq (1-\rho_{d_2})^{(q^k-1)/(q-1)}\leq \exp(-c'' \sqrt{n}/q).
$$
\end{proof}

Finally, we consider the improvement of the Gilbert--Varshamov bound implied by our argument.
We shall state the result in a probabilistic form:
\begin{cor}
  Let $q$ be a prime power and $\alpha \in (0,\frac{1}{2})$. 
  There exists constants $c,C>0$ depending on $q$ and $\alpha$ such that,
  for a sufficiently large integer $n$ and 
  $\alpha n \le d \le (1-\alpha)(1-\frac{1}{q})n$,
  with probability greater than 
  $\exp(-c\sqrt{n})$, a uniform random $\lfloor k+\frac{1}{2}\log_q(n)-C\rfloor$--dimensional linear code has the minimal distance at least $d$
  where $k$ is the largest integer such that 
  $$
    \frac{1}{q} \frac{q^n}{\sum_{j=0}^{d-1} {n \choose j}(q-1)^j} < q^{k}\le \frac{q^n}{\sum_{j=0}^{d-1} {n \choose j}(q-1)^j}.
  $$
  (i.e.\ the dimension in Gilbert--Varshamov's bound)
\end{cor}
\begin{proof}
  Notice that $k$ is the largest integer satisfying $q^k \rho_{d-1} \leq 1$.
  The Gilbert--Varshamov result states that there 
  exists a $k$--dimensional linear code with the minimal distance at least $d$.
  
  Let $t\ge 0$ be a positive integer which we will determine later.
  Further, let $w_{\min}$ be the minimal weight of $\frac{q^{k+t}-1}{q-1}$ i.i.d.\ random vectors
  uniformly distributed over $\F^n$, and let $d_{\min}$ be the minimal distance of the uniform random $(k+t)$--dimensional
  linear code in $\F^n$.
  We have 
  \begin{align*}
   \PP \{w_{\min} \ge d\}
    &= (1 - \rho_{d-1})^{\frac{q^{k+t}-1}{q-1}} \\
    &\ge \exp\Big( - 2\rho_{d-1}\frac{q^{k+t}-1}{q-1}\Big) \\
    &\geq
      \exp( - 2\rho_{d-1} q^{k+t})
     \\
    &\ge \exp(-2q^t).
  \end{align*}
  Recall the $q$-ary entropy function 
  \begin{align*}
    H_q(x) =  x \log_q(q-1) - x \log_q(x) - (1-x)\log_q(1-x)
  \end{align*}
  which appears in the Gilbert--Varshamov bound. 
  It is a monotone increasing function on $(0, 1-\frac{1}{q})$ 
  with $H_q(0)=0$ and $H_q(1)=\log_q(q-1)$. Furthermore, for $x \in (1,1-\frac{1}{q})$,
  \begin{align}
    \label{eq:qEntropyApproximation}
    H_q(x) = \frac{1}{n}\log_q\brk{\sum_{i=0}^{x n}{n \choose i}(q-1)^i} + o(1)
    =\frac{1}{n}\log_q\brk{\rho_{xn}q^{n}} + o(1)
  \end{align}
  whenever $xn$ is an integer. (See \cite[Proposition 3.3.1]{GRS2019})

  With $  q^k\rho_d \leq 1<q^{k+1}\rho_d$, we have 
  \begin{align*}
    H_q\Big(\frac{d}{n}\Big) = 1-\frac{k}{n}+o(1).    
  \end{align*}
  Therefore, there exist $0<R_1<R_2<1$ depending only on $q,\alpha$ such that 
  \begin{align*}
    R_1 \le \frac{k}{n} \le R_2.
  \end{align*}
  Now we apply Theorem \ref{main} to get 
  \begin{align*}
    \PP\set{d_{\min} \ge d + t } &\ge  \PP\set{w_{\min} \ge d + t }
    - \abs{ \PP\set{d_{\min} \ge d + t }-\PP\set{w_{\min} \ge d + t }}\\
    &\ge \exp(-2q^t)-\exp(-c_{\alpha,q}\sqrt{n})
  \end{align*}
  where $c_{\alpha,q} = c(R_1,R_2,q)$. Choosing 
  \begin{align*}
    t = \frac{1}{2}\log_q{n}+\log_q(\frac{c_{\alpha,q}}{4})
  \end{align*}
  we obtain the desired bound.

\end{proof}

It is not difficult to check that the above corollary implies Corollary~\ref{cor: mp GV} from the introduction.

\bibliographystyle{plain}

\bibliography{Ref}

\end{document}